\documentclass[a4paper]{article}
\usepackage[left=3cm,right=3cm,top=2.5cm,bottom=2.5cm]{geometry}
\usepackage{amsmath,amssymb, amsthm}
\usepackage{amsmath,amssymb}
\usepackage{graphicx}

\usepackage{url}

\renewcommand{\leq}{\leqslant}
\renewcommand{\geq}{\geqslant}
\renewcommand{\le}{\leqslant}
\renewcommand{\ge}{\geqslant}

\newcommand{\Empty}{{\tt Empty}}
\newcommand{\Oh}{\mathcal{O}}
\renewcommand{\Pr}[1]{\ensuremath{\operatorname{\mathbf{Pr}}\left[#1\right]}}

\newcommand{\Ex}[1]{\ensuremath{\operatorname{\mathbf{E}}\left[#1\right]}}
\newcommand{\Var}[1]{\ensuremath{\operatorname{\mathbf{Var}}\left[#1\right]}}

\usepackage{setspace,color}   

\newtheorem{thm}{Theorem}  

\newtheorem{lem}[thm]{Lemma}
\newtheorem{defi}{Definition}

\newtheorem{cor}[thm]{Corollary}

\usepackage{bbm}

\numberwithin{thm}{section}
\title{Balanced Allocation on Graphs:  A Random Walk Approach}
\author{Ali Pourmiri
{\small{Institute for Research in Fundamental Sciences (IPM), Tehran, Iran}}\\ 
{\small{\tt pourmiri@ipm.ir}}}

\begin{document}
\maketitle
\abstract{The standard balls-into-bins model is a process which randomly allocates  $m$  balls into $n$ bins where each ball  picks  $d$ bins independently and uniformly at random and  the ball is then  allocated in a least loaded bin in the set of $d$ choices. 
	When $m=n$ and $d=1$, it is well known that  at the end of process the maximum number of balls at any bin, the \emph {maximum load},  is  $(1+o(1))\frac{\log n}{\log\log n}$ with high probability\footnote{With high probability refers to an event that holds with probability $1-1/n^c$, where $c$ is a constant. For simplicity,  we sometimes abbreviate it  as whp. }. Azer et al. \cite{ABKU99} showed that  for the $d$-choice process, $d\ge 2$, provided ties are broken randomly, the maximum load is 
	$\frac{\log\log n}{\log d}+\Oh(1)$.
	
	In this paper we  propose   algorithms for  allocating $n$ sequential balls into   $n$ bins that are  interconnected  as a  $d$-regular $n$-vertex graph $G$, where $d\ge3$ can be any integer.
	Let $l$ be a given positive integer. In each round $t$, $1\le t\le n$, ball $t$ picks a node of $G$ uniformly at random and  performs a non-backtracking random walk of length $l$  from the chosen node.  
	Then it
	allocates itself on  one of the visited nodes with minimum load (ties are broken uniformly at random).  
	Suppose that $G$ has a sufficiently large girth and $d=\omega(\log n)$. Then we establish  an upper bound for the
	maximum number of balls at any bin after allocating $n$ balls by the algorithm, called  {\it maximum load},  in terms of $l$ with high probability. We also  show   that  the upper bound  is at most an $\Oh(\log\log n)$ factor above the lower bound that is proved for the algorithm.
	In particular, we show that  
	if
	we set $l=\lfloor(\log n)^{\frac{1+\epsilon}{2}}\rfloor$, for every constant $\epsilon\in (0, 1)$, and 
	$G$ has girth at least $\omega(l)$,  then the maximum load attained by the algorithm is bounded by  $\Oh(1/\epsilon)$ with high probability.
	Finally, we slightly modify the algorithm to have  similar results for balanced allocation on $d$-regular graph with $d\in[3, \Oh(\log n)]$ and sufficiently large girth.}

\section{Introduction}

The standard balls-into-bins model is a process which randomly allocates  $m$  balls into $n$ bins where each ball  picks  $d$ bins independently and uniformly at random and  the ball is then  allocated in a least loaded bin in the set of $d$ choices. 
When $m=n$ and $d=1$, it is well known that  at the end of process the maximum number of balls at any bin, the \emph {maximum load},  is  $(1+o(1))\frac{\log n}{\log\log n}$ with high probability. Azer et al. \cite{ABKU99} showed that  for the $d$-choice process, $d\ge 2$, provided ties are broken randomly, the maximum load is 
$\frac{\log\log n}{\log d}+\Oh(1)$.
For a complete survey on the standard balls-into-bins process  we refer  the  reader to \cite{sur01}.  
Many subsequent works consider the settings where the choice of bins are not necessarily independent and  uniform.  For instance,
V\"ocking \cite{V03} proposed an algorithm called \emph{always-go-left} that uses exponentially smaller number of choices and achieve  a  maximum load
of  $\frac{\log\log n}{d \phi_d } +\Oh(1)$ whp, where $1\le \phi_d\le 2$ is  an specified constant.
 In this algorithm, 
the bins are  partitioned  into  $d$  groups of size $n/d$ and  each ball picks one random bin from each  group. The ball is then allocated  in a least loaded bin among the chosen bins and   ties are broken  asymmetrically. 

In many applications
selecting  any  random set of choices is costly. 
For example, in  peer-to-peer or cloud-based systems  balls (jobs, items,...) and bins (servers, processors,...) are randomly placed in a metric space (e.g., $\mathbb{R}^2$) and  the balls have to be allocated on bins that are close to them as it minimizes  the  access latencies. With regard to such applications,
Byer et al. \cite{BCM04} studied a model, where $n$ bins (servers) are  uniformly at random placed on a  geometric space.  Then each ball in turn picks $d$ locations in the space and allocates itself on a nearest neighboring bin with minimum load among other $d$ bins. In this scenario, the probability that a location close to a server is chosen  depends on the 
distribution of other servers in the space and hence there is no a  uniform distribution over the potential choices. Here, the authors  showed the maximum load is $\frac{\log\log n}{\log d}+\Oh(1)$ whp. 
Later on, 
Kenthapadi and Panigrahy \cite{KP06} proposed a model in which bins are interconnected as  a $\Delta$-regular graph and  each  ball picks a random edge of the graph.  It is then  placed at one of  its endpoints  with smaller load. This allocation algorithm results in a maximum load of $\log\log n+\Oh\left(\frac{\log n}{\log (\Delta/\log^4 n)}\right)+\Oh(1)$.
Peres et al. \cite{PTW14} also considered a similar model where   number of balls $m$ can be much larger than $n$ (i.e., $m\gg n$) and the graph is not necessarily regular. Then, they established  upper bound $\Oh(\log n/\sigma)$ for the gap between the maximum and the minimum loaded bin after allocating $m$ balls, where $\sigma$ is the edge expansion of the graph.
Following the study of balls-into-bins with correlated choices,  
Godfrey \cite{God08} generalized the model introduced by  Kenthapadi and Panigrahy  such that each ball picks an random edge of a hypergraph that has  $\Omega(\log n)$ bins and satisfies  some mild conditions. Then
he showed that the maximum load is a constant whp.

Recently, Bogdan et al. \cite{BSS013} studied a model where each ball picks a random node and performs a local search from the node to find a node with local minimum load, where it is  finally placed on. They showed that when  the graph is a constant degree  expander, the local search guarantees a maximum load of $\Theta(\log\log  n)$ whp.

\paragraph{\bf Our Results.}{
	
	In this paper, we study balls-into-bins models, where each ball chooses a set of related bins. 
	We propose  allocation  algorithms for  allocating $n$ sequential balls into   $n$ bins that are  organized as a   $d$-regular $n$-vertex graph $G$.
	Let $l$ be a given positive integer.
	A non-backtracking random walk (NBRW) ${W}$    of  length $l$ started from a node is a  random  walk in $l$  steps so that in each  step the walker picks a neighbor uniformly at random and moves to that neighbor   with an additional property that the walker never traverses an edge twice in a row. Further information about NBRWs  can be found in   \cite{AL07} and \cite{AL09}. 
	Our allocation algorithm, denoted by $\mathcal{A}(G, l)$,
	is based on a random  sampling of bins from   the neighborhood of a given node  in $G$ by a NBRW from the node.
	The algorithm proceeds as follows:
	In each round  $t$, $1\le t\le n$, ball $t$ picks a node of $G$ uniformly at random and performs a 
	NBRW $W=(u_0, u_1\ldots, u_l)$, called \emph{$l$-walk}. After that  the ball allocates itself on one of the visited nodes with minimum load and ties are broken randomly.
	Our result concerns  bounding   the maximum load attained by $\mathcal{A}(G, l)$, denoted by $m^*$, in terms of $l$.
	Note that if the balls are allowed to take NBRWs of length $l=\Omega(\log n)$ on a graph with girth at least $l$, then the visited nodes by each ball generates a random hyperedge of size $l+1$. Then applying the Godfrey's result \cite{God08} implies a constant  maximum load whp. So, for the rest of the paper we  focus on  NBRWs of  sub-logarithmic length (i.e., $l=o(\log_ d n)$).
	We also assume that $l=\omega(1)$ and  $G$
	is a  $d$-regular  $n$-vertex graph with girth
	at least $\omega(l\log\log n)$ and   $d=\omega(\log n)$. However, when    \mbox{$l=\lfloor(\log n)^{\frac{1+\epsilon}{2}}\rfloor$}, for any constant $\epsilon\in(0,1)$, $G$ with girth at least $\omega(l)$ suffices  as well. 
	It is worth mentioning that there exist several explicit families of $n$-vertex $d$-regular graph with arbitrary degree $d\ge 3$ and  girth $\Omega(\log_d n)$ (see e.g. \cite{Dahan14}).
	
	In order to present the upper bound, we consider two cases:
	\begin{itemize}
		\item[I.] If  $l\ge 4\gamma_G$, where $\gamma_G=\sqrt{\log _d n}$, then 
		we show that whp, \[m^\ast=\Oh\left(\frac{\log\log n}{\log(l/\gamma_G)}\right).\] 
		Thus, for a given $G$ satisfying the girth condition,  if we set  \mbox{$l=\lfloor(\log_d n)^{\frac{1+\epsilon}{2}}\rfloor $}, for any constant $\epsilon\in(0,1)$,   then we have
		\mbox{$l/\gamma_G\ge (\log n)^{\epsilon/2}$} and  by applying the above upper bound we have $m^*=\Oh(1/\epsilon)$ whp.  
		\item[II.]
		If $\omega(1) \le l\leq 4\cdot \gamma_G $, then we show that whp, 
		\[m^*=\Oh\left(\frac{\log_d n\cdot \log\log n}{l^2}\right).\]
	\end{itemize} 
	 In addition to the upper bound, we prove that whp,
	\[ m^*=\Omega\left(\frac{\log_d n}{l^2}\right)\]
	(for a proof see Appendix \ref{App:lower}). 
	If $G$ is a $d$-regular graph with $d\in[3, \Oh(\log n)]$, then  
	we slightly modify allocation algorithm $\mathcal{A}(G,l)$ and show the similar results for $m^*$ in $l$.
	The algorithm $\mathcal{A'}(G, l)$
	for sparse graphs proceeds as follows: 
	Let us first define  parameter
	\[r_G=\lceil2\cdot\log_{d-1}\log n\rceil.
	\] 
	For each ball $t$, the ball takes a NBRW of size $l\cdot r_G$, say 
	$(u_0, u_1,\cdots,u_{lr_G})$, and then
	a subset of visited nodes,   {\mbox{$\{u_{j\cdot r_G} ~|~0\le  j\le  l\}$}}, called {\it potential choices},
	is  selected and finally the ball is allocated on a least-loaded node of potential choices (ties are broken randomly). Provided $G$ has sufficiently large girth, we show the similar upper and lower bounds as the allocation algorithm $\mathcal{A}(G,l)$ on $d$-regular graphs with $d=\omega(\log n)$ (see Appendix \ref{sec:sparse} ). 
}



\paragraph{\bf Comparison with Related Works.}{The setting of our work is closely related to \cite{BSS013}. In this paper in each step a ball picks a node of a graph uniformly at random and performs a  local search to find a node with local minimum  load and finally allocates itself on  it. They showed that  with high probability the local search on expander graphs obtains a maximum load of $\Theta(\log\log n)$.     In comparison to the mentioned result, our new protocol achieves a further reduction in the maximum load, while still allocating a ball close to its origin.  
	Our result suggests a trade off between allocation time and maximum load.
	In fact we  show a constant upper bound for  sufficient long walks (i.e., $l=(\log n) ^{\frac{1+\epsilon}{2}} $, for any constant $\epsilon \in (0,1)$). 
	Our work  can also be related to the one by 
	Kenthapadi and Panigrahy where each ball picks a random edge from a $n^{\Omega(1/\log\log n)}$-regular graph and  places itself   on one of the endpoints of the edge with smaller load. This model results into a maximum load of $\Theta(\log\log n)$.
	Godfrey   \cite{God08}  considered  balanced allocation on hypergraphs where  balls choose a  random edge  $e$ of a hypergraph satisfying  some conditions, that is, first the size of each edge $s$ is $\Omega(\log n)$ and  $\Pr{u\in e}=\Theta(\frac{s}{n})$ for any bin $u$.  The latter one is called \emph{balanced condition}.
	Berenbrink et al.  \cite{BBFN12} simplified  Godfrey's proof and slightly  weakened  the balanced condition but since   both analysis apply a Chernoff bound, it seems unlikely that one can extend the analysis for hyperedges of size $o(\log n)$. 
	Our model can also be viewed as
	a balanced allocation on hypergraphs,  because  every $l$-walk is  a random hyperedge of size $l+1$ that also satisfies the balanced condition (see Lemma \ref{first}). By setting the right parameter for $l=o(\log n)$, we show  that the algorithm achieves a constant maximum load with sub-logarithmic number of choices.

	In a different context, Alon and Lubetzky  \cite{AL09} showed that  if a particle starts a NBRW of length $n$ on  $n$-vertex graph with high-girth then the number of  visits to nodes has a Poisson distribution. In particular they showed that   the maximum visit to a node is at most $(1+o(1))\cdot\frac{\log n}{\log\log n}$. Our result can be also seen as an application of the mathematical concept of NBRWs to task allocation in distributed networks.
}
\paragraph{\bf Techniques.}{ 
	To derive a lower bound for the maximum  load  we first show that  whp there is a path of length $l$ which is traversed  by at least $\Omega\left({\log_d n}/{l}\right)$ balls. Also, each path  contains $l+1$  choices and hence, by pigeonhole principle there is a node with 
	load at least  $\Omega\left({\log_d n}/{l^2}\right)$, which is a lower bound for $m^*$. 
	We establish the upper bound based  on \emph{witness graph} techniques.
	In our model,
	the potential  choices for each ball are highly correlated, so the technique for building the witness graph is somewhat different from the one for standard balls-into-bins.
	Here we propose a new approach for constructing  the witness graph.
	We also show a  key 
	property of the algorithm,  called  
	\emph{$(\alpha, n_1)$-uniformity}, that is useful for our proof technique. 
	We say an allocation algorithm is $(\alpha, n_1)$-uniform if 
	the
	probability that, for every   $1\le t\le n_1$, ball $t$  is placed on an arbitrary node is bounded by $\alpha/n$,
	where $n_1=\Theta(n)$ and  $\alpha=\Oh(1)$.	Using this property we conclude that  
	for a given set of nodes of size $\Omega(\log n)$, after allocating $n_1$ balls,  the average load of nodes in the set is some constant  whp.  Using   {witness graph} method we show that if there is a node with load larger than some threshold   then there is a collection of nodes of  size $\Omega(\log n)$ where each of them has load larger than some specified  constant. Putting these together implies that 
	after allocating $n_1$ balls the maximum load, say $m_1^*$, is bounded as required  whp. To derive an upper bound for the maximum load after allocation $n$ balls,  we divide the allocation process into $n/n_1$ phases and  show that  the maximum load at the end of each phase increases by at most $m_1^*$ and hence $m^*\le (n/n_1)m_1^*$ whp. 

}		

\paragraph{\bf Discussion and Open Problems.}{
		In this paper, we proposed  balls-into-bins model, where
		each ball picks a set of nodes that are visited by a NBRW of length $l$ and place itself on a visited node with minimum load. One may ask whether  it is possible to replace a NBRW of length $l$ by  several parallel random walks of shorter length (started from the same node) and get the  similar results?
	 
	 In our result we constantly use the assumption that the graph locally looks like a $d$-ary tree.
		 It is also known that   cycles in random regular graph are restively far from each other (e.g, see \cite{CooperFR09}), so  we believe  that  our approach can be extended  for balanced allocation on random regular graphs.

		 Many works in this area (see e.g.\cite{BSS013,KP06}) assumed that  the underlying networks   is regular, it would be interesting  to investigate random walk-based algorithms for irregular graphs.
			
			}	

\paragraph{\bf Outline.}{In Section \ref{sec:defi}, we present notations and some preliminary results that are required for the analysis of the algorithm. In Section \ref{sec:witness} we show how to construct a witness graph and then in Section \ref{main thm} by applying the results we the upper bound for the maximum load.
	}
\section{Notations, Definitions and Preliminaries}\label{sec:defi}
	In this section we provide  notations, definitions and  some preliminary results. 
	A {\emph{non-backtracking random walk}} (NBRW) ${W}$    of  length $l$ started from a node is a simple  random  walk in $l$  steps so that in each  step the walker picks a neighbor uniformly at random and moves to that neighbor   with an additional property that the walker never traverses an edge twice in a row. Throughout  this paper we assume that  $l\in [\omega(1), o(\log_d n)]$ is a given parameter and $G$ is a $d$-regular graph with  girth $10\cdot l\cdot \log \log n$. Note that  we will see that the  condition on  the girth can be relaxed to $\omega(l)$, for any $l$ higher than ${(\log_d n)^{\frac{1+\epsilon}{2}}}$, where $\epsilon\in(0, 1)$ is a
	constant.  
	
	It is easy to see that  the visited nodes by a non-backtracking walk of length $l$  on $G$ induces a path of length $l$, which is called an {\it $l$-walk}. 
	For simplicity, we use $W$ to denote both the $l$-walk and the set of visited nodes by the $l$-walk.  
	Also, we define $f(W)$ to be  the number of balls in a least-loaded node of $W$. 
	The {\it height } of  a ball allocated on a node is the number balls that are placed on the node  before the ball.
	
	For every two nodes $u, v\in V(G)$, let $d(u, v)$ denote the length of shortest path  between $u$ and $v$ in $G$.
	Since $G$ has girth at least $\omega(l)$, every path of length at most $l$ is  specified by its endpoints, say $u$ and $v$. So  we denote the path  by interval $[u, v]$. 
	Note that for any graph $H$, $V(H)$ denotes the vertex set of $H$.

	\begin{defi}[Interference Graph]For every given pair $(G, l)$, the interference graph $ \mathcal{I}({G, l})$ is defined as follows:  The vertex set of   $ \mathcal{I}({G, l})$ is the set of all $l$-walks  in $G$
		and two vertices $W$ and $W'$ of  $ \mathcal{I}({G, l})$  are connected if and only if  $W\cap W'\neq \emptyset$.
		Note that if  pair $(G, l)$ is clear from the context, then the interference graph is denoted by  $\mathcal{I}$.
	\end{defi} 
	Now, let us  interpret 
	allocation process $\mathcal{A}(G, l)$ as follows: 
	
	For every ball $1\le t\le n$, the algorithm picks a vertex of $\mathcal{I}(G, l)$, say $W_t$, uniformly at random and then allocates ball $t$ on a least-loaded node of $W_t$ (ties are broken randomly).
	Let $1\le n_1 \leq n$ be a given integer and assume that $\mathcal{A}(G, l)$  has allocated balls   until the $n_1$-th ball.
	We then define  $\mathcal{H}_{n_1}(G, l)$ to be the induced subgraph of  $\mathcal{I}(G, l)$ by 
	\mbox{ $\{W_t : 1\le t \le n_1\}\subset V(\mathcal{I})$.} 
	\begin{defi}
		Let $\lambda$ and  $\mu$ be given positive integers. We say rooted tree $T\subset \mathcal{I}(G, l)$ is a $(\lambda, \mu)$-tree if $T$ satisfies:
		\begin{itemize}
			\item[1)] $|V(T)|=\lambda$,
			\item[2)] $|\cup_{W\in V(T)} W|\ge \mu$.
		\end{itemize}
		Note that the latter condition is well-defined because every vertex of $T$ is an $(l+1)$-element subset of $V(G)$. 
		A $(\lambda, \mu)$-tree $T$ is called $c$-loaded, if  $T$ is contained in  $\mathcal{H}_{n_1}(G, l)$, for some $1\le n_1\le n$, and every node in $\cup_{W\in V(T)}W$ has load at least $c$.
	\end{defi}

	\subsection{Appearance Probability of a $c$-Loaded $(\lambda, \mu)$-Tree }
	In this subsection we formally define the notion of $(\alpha, n_1)$-uniformity for allocation algorithms,  and then present our key lemma concerning the uniformity of $\mathcal{A}(G, l)$. By using this lemma  we establish an upper bound for the probability that a 
	$c$-loaded $(\lambda, \mu)$-tree contained in $\mathcal{H}_{n_1}$ exists.
	The proof of the following lemmas can be found in Appendix \ref{Appendix-1}.
	\begin{defi}
		Suppose that $\mathcal{B}$ be an  algorithm that allocates  $n$ sequential  balls into $n$ bins. 
		Then we say $\mathcal{B}$ is {\it $(\alpha, n_1)$-uniform} if,
		for every $1\le t\le n_1$ and   every bin $u$, after allocating $t$ balls we have that 
		\[
		\Pr{ \text{ball $t+1$ is allocated on $u$ }}\leq \frac{\alpha}{n},
		\]
		where $\alpha$ is some constant and $n_1=\Theta(n)$. 
	\end{defi}
	\begin{lem}[Key Lemma]\label{lem:uniform}
		$\mathcal{A}(G, l)$~is an  $(\alpha, n_1)$-uniform allocation algorithm, where  \mbox{$n_1=\lfloor n/(6\mathrm{e}\alpha)  \rfloor$}.
		\end{lem}
			In the next lemma, we derive an upper  bound for  the appearance probability a $c$-loaded $(\lambda, \mu)$-tree, whose proof  is inspired by \cite[Lemma 2.1]{KP06}.
	\begin{lem}\label{lem:tree}
		Let  $\lambda$, $\mu$ and $c$ be positive integers.
		Then the probability that there exists  a $c$-loaded $(\lambda, \mu)$-tree contained in $ \mathcal{H}_{n_1}(G, l)$ 		is at most
		\[
		n\cdot\exp({4\lambda\log(l+1) -c\mu}).
		\]

	\end{lem}

\section{ Witness Graph}\label{sec:witness}

In this section, we show that if there is a node whose load is larger than a threshold,  then we can construct a $c$-loaded $(\lambda, \mu)$-tree contained in $\mathcal{H}_{n_1}(G, l)$.
   Our construction is based on an  iterative application of a $2$-step procedure,  called {\sf Partition-Branch}.
   	Before we explain the construction, we draw the reader's attention  to the following remark:
   	\paragraph{Remark.}{ 
   		The intersection (union)   of two arbitrary graphs is a graph whose vertex set  and edge set are  the intersection (union) of the vertex  and edge sets of those graphs. Let  $\cap_g$ and $\cup_g$ denote the graphical intersection and union.
   		Note that we use $\cap$ $(\cup)$ to denote the set intersection (union) operation.
   		Moreover, since $G$ has girth $\omega(l)$, the graphical intersection of every two $l$-walks in  $G$ is either empty or a path (of length $\le l$). Recall that $W$ denotes both an $l$-walk and the set of nodes in the $l$-walk.}
 
 	\begin{figure}[t]
 		\vspace{0cm}
 		\centering
 		\includegraphics[width=1\textwidth]{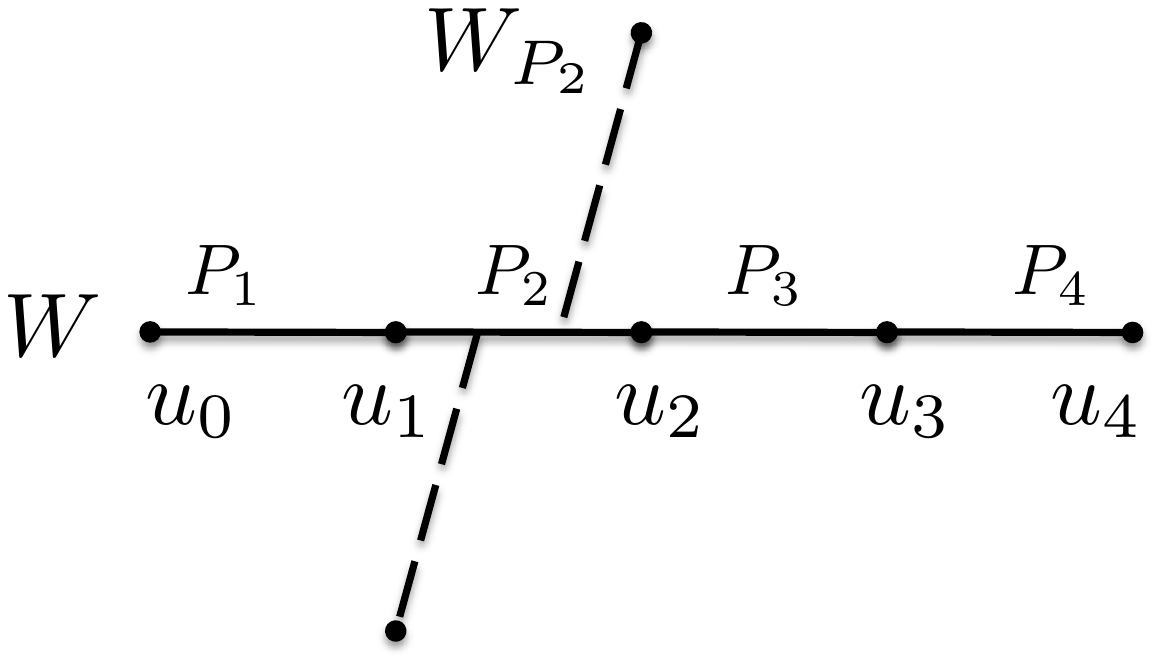}
 		\vspace{-6cm}
 		\caption{\small{\small The {\sf Partition} step on $W$ 
 				for $k=4$ and  the {\sf Branch} step for $P_2$ 
 				that gives $W_{P_2}$, shown by dashed line.}\label{fig1} }
 	\end{figure}

\paragraph{Partition-Branch.}
	{
		 Let  $k\ge 1$ and $\rho\ge 1$ be given integers and $W$ be an   $l$-walk
	  with $f(W)\ge \rho+1$. The
		{\sf Partition-Branch} procedure on $W$ with parameters $\rho$ and $k$, denoted by $PB(\rho, k)$,  proceeds as follows:
		\begin{itemize}
			\item[]{\sf Partition}: It   partitions $W$   into $k$ edge-disjoint subpaths:
			\[ 
			\mathcal{P}_k(W)=\{[u_{i}, u_{i+1}] \subset W, 0\le i \le k-1\},
			\]
			where $d(u_{i}, u_{i+1})\in \{ \lfloor l/k\rfloor, \lceil l/k\rceil  \}$. 
			\item[]{\sf Branch}: For a given $P_i=[u_{i}, u_{i+1}] \in\mathcal{P}_k(W)$, it finds  (if exists) another $l$-walk $W_{P_i}$ intersecting $P_i$
		that   satisfies  the following  conditions:
			\begin{itemize}
				\item[(C1)]$\emptyset\neq W_{P_i}\cap W\subseteq P_i\setminus\{u_i,u_{i+1}\}$.
				\item[(C2)] ${f}(W_{P_i})\ge f(W)-\rho$.
				
			\end{itemize}
		\end{itemize}
	}

		We say procedure $PB(\rho, k)$ on a given  $l$-walk $W$ is {\it valid}, if 
		for every $P\in \mathcal{P}_k(W)$, $W_P$   exists. We usually refer to  $W$ as the father of $W_P$.
		For a graphical view of the {\sf Partition-Branch} procedure  see 
		Figure \ref{fig1}.
	
		\begin{defi}[Event $\mathcal{N}_\delta$]
			For any given $1\le \delta\le l$, we say that  event  $\mathcal{N}_\delta$ holds,  if after allocating at most $n$ balls by $\mathcal{A}(G, l)$, every path of length $\delta$ is contained in less than $6\log_{d-1}n/\delta$ $l$-walks that are randomly chosen by $\mathcal{A}(G, l)$.
			
		\end{defi} 
		
		For the sake of construction, let us  define a set of parameters, depending on $d$, $n$, and $l$, which are used throughout the paper
		\begin{align*}
		k&:=\max\{4, \lfloor l/\sqrt{\log_d n}\rfloor \},\\
		\delta&:=\lfloor\lfloor l/k\rfloor/4\rfloor,\\
		\rho&:= \lceil 6 \log_d n /\delta^2\rceil.
		\end{align*}

		\begin{lem}\label{lem:PB}
			Suppose that event $\mathcal{N}_\delta$ holds and  $W$ be an $l$-walk with 
			\[f(W)\ge \rho+1.
			\]
			Then 
			the procedure $PB(\rho, k)$ on $W$ is valid. 
		\end{lem}	
		For a proof see Appendix \ref{Appendix-2}.

		\subsection{Construction of Witness Graph}
			In this subsection,  we show how to construct a $c$-loaded $(\lambda, \mu)$-tree contained in $\mathcal{H}_{n_1}$. 
				Let $U_{n_1, l, h}$ denote the event that   after allocating  at most $n_1\le n$ balls by $\mathcal{A}(G, l)$ there is a node with load at least  $h \rho +c+1$, where $c=\Oh(1)$ and $h=\Oh(\log\log n)$ are  positive integers that  will be fixed later.
			Suppose that 
			event  $U_{n_1, l, h}$ conditioning on $\mathcal{N}_\delta$
			happens. Then there is an $l$-walk $R$, called root,  that corresponds  to the ball at height   $h \rho +c$ and has \mbox{$f(R)\ge h \rho +c $}.
			Applying Lemma \ref{lem:PB} shows  that  $PB(\rho, k)$  on $R$ is valid.
			So, let us define 
			\[
			\mathcal{L}_1:=\{W_P, P\in \mathcal{P}_k(R)\},
			\] which is called the first level  and $R$ is the father of all $l$-walks in $\mathcal{L}_1$.
			(C2) in the {\sf Partition-Branch} procedure  ensures that  for every $W\in \mathcal{L}_1$, 
			\[
			f(W)\ge (h-1)\rho+c.
			\] 
			
			\begin{figure}[t]
				\vspace{0cm}
				\centering
				\includegraphics[width=1\textwidth]{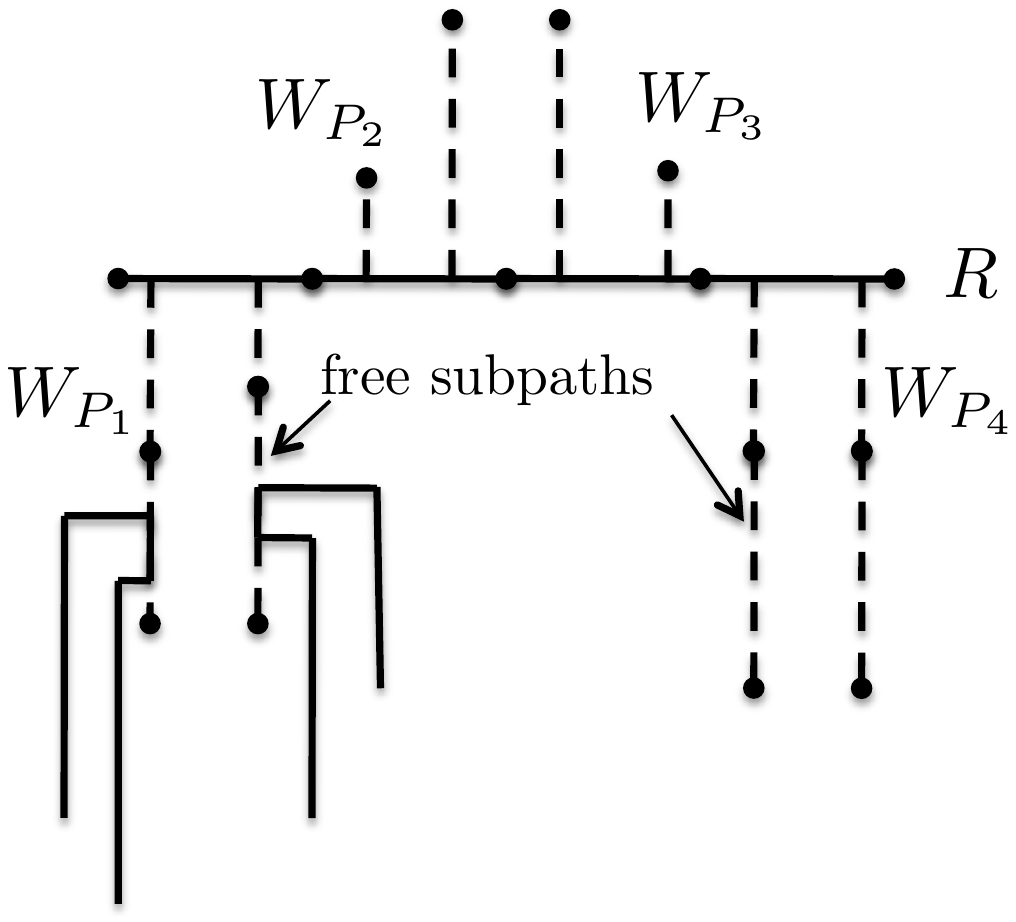}
				\vspace{-4.5cm}
				\caption{\small The first level
					$\mathcal{L}_	1=\{W_{P_1}, W_{P_{2}}, W_{P_3}, W_{P_4}\}$ and  the {\sf Branch} step for free subpaths of $\mathcal{P}_k(W_{P_1})$.\label{fig2}}
				
			\end{figure}
			Once we have the first level we recursively build the $i$-th level from the $(i-1)$-th level, for every $2\le i \le h$.
			We know that each $W$ except $R$ is created by the {\sf Branch} step on its father.
			Let us fix  $W\in \mathcal{L}_{i-1}$ and its father $W'$. We then apply the {\sf Partition} step  on  $W$  and get $\mathcal{P}_k(W)$.
			We say  $P\in \mathcal{P}_k(W)$ is a {\it free} subpath  if it does not share any node with $W'$. By (C1),  we have that  $\emptyset\neq W\cap W'=[u, v]\subset P'$, for some $P'\in \mathcal{P}_k(W')$ and hence
			$d(u, v)\le \lceil l/k\rceil$.  So,  $[u, v]$
			shares node(s) with   at most  $2$  subpaths in
			$\mathcal{P}_k(W)$ and  thus $\mathcal{P}_k(W)$ contains at least $k-2$ free subpaths.
			Let  $\mathcal{P}^0_k(W)\subset \mathcal{P}_k(W)$ denote
			an arbitrary  set of free subpaths of size $k-2$.  By (C2) and the recursive construction,  we have that  $f(W)\ge (h-i+1)\rho +c$, 
			for each $W\in\mathcal{L}_{i-1}$. Therefore, by Lemma    \ref{lem:PB}, $PB(\rho, k)$ on $W$ is valid. 
			Now we define the $i$-th level as follows,
			
			\begin{align*}
			\mathcal{L}_i=\bigcup_{W\in \mathcal{L}_{i-1}}\{W_P, P\in \mathcal{P}^0_k(W)\}.
			\end{align*}
			For a graphical view see Figure \ref{fig2}. The following lemma guarantees that our construction gives a $c$-loaded $(\lambda, \mu)$-tree in $\mathcal{H}_{n_1}$ with desired parameters (for a proof see Appendix \ref{app:const}).

			\begin{lem}\label{lem:witness}
				Suppose that $G$ has girth at least $10hl$ and $U_{ n_1, l, h}$ conditioning on $\mathcal{N}_\delta$ happens.
				Then  there exists a $c$-loaded  $(\lambda, \mu)$-tree $T\subset \mathcal{H}_{n_1}$, where
				$\lambda=1+ k\sum_{j=0}^{h-1} (k-2)^j$ and $\mu=(l+1)\cdot k(k-2)^{h-1}$. \end{lem}

	\section{Balanced Allocation on Dense Graphs }\label{main thm}
		In this section we show the upper bound for the maximum load attained by $\mathcal{A}(G, l)$ for $d$-regular graph with $d=\omega(\log n)$. Let us  
	recall the  set of parameters  for given $G$ and $l$ as follows,
			\begin{align*}
				k&:=\max\{4, \lfloor l/\sqrt{ \log_d n}\rfloor \},\\
				\delta&:=\lfloor\lfloor l/k\rfloor/4\rfloor,\\
				\rho&:= \lceil 8 \log_d n /\delta^2\rceil,
			\end{align*}
			and $U_{n_1, l, h}$ is the event that at the end of round $n_1$,  there is a nodes  with load at least  $h\rho+c+1$, where $c$ is a constant and
			\begin{align*}
			h&:=\left\lceil\frac{ \log\log n}{\log (k-2)}\right\rceil.
			\end{align*}
			Note that when $l=(\log n)^{\frac{1+\epsilon}{2}}$ with constant $\epsilon\in (0, 1)$, then \[
			k=\lfloor l/\sqrt{ \log_d n}\rfloor\ge  l/\sqrt{ \log_3 n}\ge (\log n)^{\epsilon/3}.
			\]  Thus, $h=\left\lceil\frac{ \log\log n}{\log (k-2)}\right\rceil$ is a constant. Therefore, in order to apply Lemma \ref{lem:witness} for this case, it is sufficient that $G$ has girth at least $10 hl$ or $\omega(l)$. 
			Also we have the following useful lemma whose proof appears  in Appendix \ref{Max-1}.
		
		\begin{lem}\label{lem:Max}
			With probability $1-o(1/n)$, $\mathcal{N}_\delta$ holds.
		\end{lem}

		\begin{thm} Suppose that $G$ is a $d$-regular graph with girth at least $10hl$ and $d=\omega(\log n)$. Then, 
			with high probability 
			the maximum load attained by $\mathcal{A}(G, l)$, denoted by $m^*$,  is bounded from above as follows:
			\begin{itemize}
				\item[I.] If $\omega(1)\le  l\le 4\gamma_G$, where $\gamma_G=\sqrt{\log _d n}$. Then we have
				\[
				m^*=\Oh\left(\frac{ \log_d n\cdot \log\log n}{l^2}\right).\]
				
				\item[II.] If  $l\ge 4\gamma_G$, then we have 
				\[
				m^*=\Oh\left(\frac{\log\log n}{\log (l/\gamma_G)}\right).
				\]
			\end{itemize}
		\end{thm}
		Note that when $l=\Theta(\gamma_G)$, we get the maximum load $\Oh(\log\log n)$.
		\begin{proof}
			By Lemma \ref{lem:uniform} we have that $\mathcal{A}(G, l)$ is an $(\alpha, n_1)$-uniform, where $n_1=\lfloor n/(6\mathrm{e}\alpha)\rfloor$. 
			Let us divide the allocation process into $s$ phases, 
			where $s$ is  the smallest integer satisfying $sn_1\ge n$. We now focus on the maximum load attained 
			by $\mathcal{A}$
			after allocating $n_1$ balls in the first phase, which is  denoted by $m_1^*$. Let us assume that $U_{n_1, l, h}$ happens.
			Now, in order  to apply Lemma \ref{lem:witness}, we only need that $G$ has girth at least $10hl$. 
			By Lemma \ref{lem:witness}, if $U_{n_1, l, h}$  conditioning on $\mathcal{N}_\delta$ happens, then 
			there is a $c$-loaded $(\lambda, \mu)$-tree $T$ contained in $\mathcal{H}_{n_1}$,
			where
			$\lambda=1+ k\sum_{j=0}^{h-1} (k-2)^j$ and $\mu=(l+1)\cdot k(k-2)^{h-1}$. Thus,  we get
			\begin{align*}
			\Pr{U_{n_1, l, h} ~|~ \mathcal{N}_\delta}\Pr{\mathcal{N}_\delta}&\leq \Pr{T \text{ exists} ~| ~\mathcal{N}_\delta}\Pr{\mathcal{N}_\delta}\\
			&=\Pr{T \text{ exists and~} \mathcal{N}_\delta}\\
			&\leq  \Pr{T \text{ exists}}.
			\end{align*}
			Therefore using the law of total probability and the above inequality we have 
			\begin{align}
			\Pr{U_{n_1, l, h}}&=\Pr{U_{n_1, l, h} ~|~ \mathcal{N}_\delta}\Pr{\mathcal{N}_\delta}+ \Pr{U_{n_1, l, h} ~|~\neg \mathcal{N}_\delta}\Pr{\neg\mathcal{N}_\delta\notag}\\
			& \le  \Pr{T \text{ exists}}+ \Pr{\neg\mathcal{N}_\delta}\notag\\
			&=  \Pr{T \text{ exists}}+ o(1/n)\label{final}.
			\end{align}
			where the last inequality follows from  $\Pr{\neg\mathcal{N}_\delta}=o(1/n)$ by Lemma \ref{lem:Max}.
			By definition of $h$, we get
			\[
			\lambda=1+k(1+(k-2)^h)\le 2k\log n\]
			and \[\mu=(l+1)k(k-2)^{h-1}\ge (l+1)(k-2)^{h}\ge (l+1)\log n.\]
			It only remains to bound  $\Pr{ T ~\text{exists}}$. By   applying  Lemma \ref{lem:tree} and substituting  $\mu$ and $\lambda$, we conclude that 
			\[
			\Pr{T~\text{exists} }\le n\exp(4\lambda\log (l+1))-c\mu)\leq n\exp\{
			-z\log n\},
			\]		
			where $z=c(l+1)-8k\log (l+1)$. Depending on $k$ we consider two cases:
			First, $k=4$. Then it is easy to see there exists a  constant $c$ such that $z\ge 2$.
			Second, $k=\lfloor l/\gamma_G\rfloor$.
			We know that $l<\log_d n$, so we have   $l\le \gamma_G^2$ and hence,
			\[
			z\ge cl-8l\log l/\gamma_G\ge l(c -16\log \gamma_G /\gamma_G)=l(c-o(1)).
			\]

			This yields  that for some integer $c>0$, $z=l(c-o(1))>2$  and hence in both cases we get $\Pr{ T ~\text{exists}}=o(1/n)$. Now, by Inequality (\ref{final})
			we infer that    $m^*_1\le h\rho +c+1$ with probability $1-o(1/n)$. 
			In what follows we show the sub-additivity of the algorithm and concludes that  in the second phase the maximum load increases by at most $m_1^*$ whp.
			Assume that we have a copy of $G$, say $G'$,  whose nodes have load exactly $m_1^*$.
			Let us consider the allocation process of a pair of balls $(n_1+t, t)$, for every $0\le t\le n_1$, by $\mathcal{A}(G, l)$ and  $\mathcal{A}(G', l)$.
			Let $X_u^{n_1+t}$ and $Y_u^t$, $t\ge 0$ denote the load of $u\in V(G)=V(G')$ after allocating balls $n_1+t$ and $t$  by $\mathcal{A}(G, l)$ and $\mathcal{A}(G', l)$, respectively.
			Now we show that for every integer $0\le t\le n_1$ and $u\in V(G)$ we have that 
			\begin{align}\label{domin}
			X_u^{n_1+t}\le Y_u^{t}.
			\end{align}
			When $t=0$, clearly the inequality holds because $Y_u^0=m_1^*$.
			We couple the both allocation processes 
			$\mathcal{A}(G, l)$ and $\mathcal{A}(G',l)$ for a given pair of balls $(n_1+t, t)$, $t\ge0$, as follows.
			For every $1\le t\le n_1,$  
			the coupled process first 
			picks 
			a one-to-one labeling function $\sigma_t: V(G)\rightarrow \{1, 2,\ldots,n\}$ uniformly at random.
			(Note that $\sigma_t$ is also defined for $G'$ as $V(G)=V(G')$.) Then it 
			applies $\mathcal{A}(G, l)$ 
			and selects  $l$-walks $W_{n_1+t}$ and its copy, say  $W'_t$, in $G'$. After that,  balls $n_1+t$ and $t$ are allocated on  least loaded nodes of $W_{n_1+t}$ and $W'_t$, respectively, and  ties are broken in favor of nodes with minimum label.
			It is easily checked that the defined process is a coupling.
			Let us  assume that  Inequality (\ref{domin})  holds for every $t_0\le t$,  then we show it for $t+1$. Let  $v\in W_{n_1+t+1}$ and $v'\in W_{t+1}'$ denote the  nodes that are the destinations  of pair $(n_1+t+1, t+1)$. Now we consider two cases:
			\begin{itemize} 
				\item[1.] $X_v^{n_1+t}<Y_v^t$.  Then allocating  ball $n_1+t+1$ on $v$  implies that  
				\[
				X_v^{n_1+t}+1=X_v^{n_1+t+1}\le Y_v^{t}\le Y_v^{t+1}.
				\]
				So, Inequality (\ref{domin}) holds for $t+1$ and every $u\in V(G)$.
				\item[2.] $X_v^{n_1+t}=Y_v^t$.
				Since $W_{n_1+t+1}=W'_{t+1}$, $v\in W_{t+1}'$ and $v'\in W_{n_1+t+1}$.
				Also we know that $v$ and $v'$  are nodes   with minimum load contained in $W_{n+t+1}$ and $ W_{t+1}$,  So we have,
				\[
				X_{v}^{ n_1+t} \le X_{v'}^{ n_1+t} \leq Y^t_{v'}\le Y^t_{v}. 
				\]
				Since $Y_v^{t}=X_v^{n_1+t}$, we have 
				\[
				Y_{v'}^t=Y_v^t=X_{v}^{n_1+t}.
				\] If  $v\neq v'$ and  $\sigma_{t+1}(v')< \sigma_{t+1}(v)$, then it contradicts  the fact that ball $n_1+t+1$ is allocated on $v$. Similarly, if $\sigma_{t+1}(v')> \sigma_{t+1}(v)$, it contradicts  that ball $t$ is allocated on $v'$. So, we have $v=v'$ and 
				
				\[
				X_{v}^{n_1+t}+1=X_v^{n+t+1}=Y^t_v+1=Y_v^{t+1}.
				\]
				
			\end{itemize}
			So in both cases,  Inequality (\ref{domin}) holds for every $t\ge 0$.
			If we set \mbox{$t=n_1$}, then  the maximum load attained by $A(G', l)$ is at most $2m^*_1$ whp.
			Therefore, by Inequality (\ref{domin}), $2m^*_1$ is an upper bound for the maximum load attained by $\mathcal{A}(G, l)$ in the second phase as well. Similarly, we  apply the union bound and conclude that after allocating the balls in  $s$ phases, the maximum load $m^*$ is at most $sm_1^*$ with probability $1-o(s/n)=1-o(1/n)$.
		\end{proof}

			\paragraph{Acknowledgment. }{The author wants to thank Thomas Sauerwald for introducing the problem and several helpful discussions. }

\bibliographystyle{plain}
\bibliography{diss1}

\appendix

\section{Omitted Proofs of Section \ref{sec:defi}}\label{Appendix-1}
In this section we show some useful results about  interference graph $\mathcal{I}(G, l)$ and present the omitted proofs of Section \ref{sec:defi}. 

	\begin{lem}\label{int}
		Suppose that  $V(\mathcal{I})$ and
		$\Delta(\mathcal{I})$ denote the vertex set and  the maximum degree of  interference graph $\mathcal{I}(G, l)$, respectively. Then we have,
		\begin{itemize}
			\item[(i)] $|V(\mathcal{I})|=nd(d-1)^{l-1}/2$,
			\item[(ii)] $\Delta(\mathcal{I})\le (l+1)^2d(d-1)^{l-1}$.
		\end{itemize}
		
		Furthermore, the number of rooted $\lambda$-vertex trees contained in $\mathcal{I}$ is bounded by  
		$ 4^\lambda\cdot|V(\mathcal{I})|\cdot\Delta(\mathcal{I})^{\lambda-1}$.
	\end{lem}
	\begin{proof}
		It is easy to see that in a graph with girth at least $\omega(l)$, the number of $l$-walks  is exactly $nd(d-1)^{l-1}/2$, (without ordering)
		which is the size of $V(\mathcal{I})$.   Since the graph locally looks like a $d$-ary tree,  the
		total number of $l$-walks including  $v$ as $j$-th visited  node  is at most
		\[d(d-1)^{j-2} (d-1)^{l-j-1}=d(d-1)^{l-1}.\] Index $j$ varies from $0$ to $l$,  so $v$ can be an element of  at most $(l+1)d(d-1)^{l-1}$ $l$-walks. Also, every $l$-walk contains $l+1$ elements and hence  every $l$-walk intersects at most $(l+1)^2d(d-1)^{l-1}$ other $l$-walks. Thus we get
		\begin{align*}
		\Delta(\mathcal{I})\leq  (l+1)^2\cdot d(d-1)^{l-1}.
		\end{align*}
		Let us now bound  the total number of rooted $\lambda$-vertex trees contained in  $\mathcal{I}$.
		It was shown  that the total number of different shape rooted trees on $\lambda$ vertices is $4^\lambda$ (For example see \cite{Knuth}); we say two rooted trees have different shapes if they are not isomorphic. 
		For any given shape, there are $|V(\mathcal{I})|$ ways to choose the root. As soon as the root is fixed, each vertex in the  first level can be chosen in at most $\Delta(\mathcal{I})$ many ways. By selecting   the vertices of the tree level by level we have that  each vertex except the root can be chosen in at most $\Delta(\mathcal{I})$ ways. So the total number of 
		rooted $\lambda$-vertex  trees in $\mathcal{I}$ is bounded by 
		\[
		4^\lambda\cdot |V(\mathcal{I})|\cdot \Delta(\mathcal{I})^{\lambda-1}. 
		\]
		
	\end{proof}

	\begin{cor}\label{cor:count}
		The size of family of $(\lambda, \mu)$-trees is bounded by
		$4^\lambda|V(\mathcal{I})|\Delta(I)^{\lambda-1}$.
	\end{cor}
	\begin{proof}
		We know that 
		every $(\lambda, \mu)$-tree $T$ is a rooted $\lambda$-vertex subtree of $\mathcal{I}$ with the additional property that  
		$|\cup_{W\in V(T)} W|\ge \mu$. This implies that the size of family of 
		rooted $\lambda$-vertex subtrees of $\mathcal{I}$ is an upper bound for the size of family of  $(\lambda, \mu)$-trees and hence  by applying  Lemma \ref{int}, we reach the upper bound 
		$4^\lambda |V(\mathcal{I})|\Delta(I)^{\lambda-1}$.
	\end{proof}

\subsection{Proof of the Key Lemma} \label{sec:uniform}
	In this subsection  we first present several  useful lemmas and then show the key lemma \ref{lem:uniform}. Before that let us define some notations.
	For every $S\subseteq V(G)$, $\Empty_t(S)$ denotes the number of empty nodes contained in $S$ after allocating $t$ balls. Let  $N(v)$ denote the set of neighbors of $v$.   
	Note that  
	to avoid a lengthy case analysis   we do not optimize the constants.
	\begin{lem}[Deviation bounds for moderate independency]\label{mod-chernoff}
		Let $X_1,\cdots, X_n$ be arbitrary binary random variables. Let
		$X_1^*, X_2^*, \cdots, X_n^*$
		be binary random variables that are mutually independent and
		such that for all $i$, $X_i^*$,
		is independent of $X_1, \cdots, X_{i-1}$. Assume that for all
		$i$ and all $x_1, . . . , x_{i-1} \in \{0, 1\}$,
		\[
		\Pr{X_i = 1|X_1= x_1,\cdots, X_{i-1}= x_{i-1}}\ge \Pr{X^*_i= 1}.
		\]
		Then for all $a\ge 0$, we have
		\[
		\Pr{ \sum_{i=1}^nX_i\le a}\le \Pr{\sum_{i=1}^n X^*_i\le a}
		\]
		and the latter term can be bounded by any deviation bound for independent
		random variables.
	\end{lem}
	The proof of the above lemma 
	can be found in \cite[Lemma 1.18]{doerrbook}.
	\begin{lem}\label{first}
		Suppose that $\mathcal{A}(G, l)$ has allocated the balls until the \mbox{$(t+1)$-th} ball, for some $0\le t\le n$. Then, for every given $v\in V(G)$ we have
		\[
		\Pr{v\in W_{t+1}}=\sum_{i=0}^{l}\Pr{C_i}=(l+1)/n,
		\]
		where ${C}_i$, $0\leq i\leq l$, is  the event that for $ W_{t+1}=(u_0,u_1,\ldots, u_{l})$, we have $v=u_{i}$. Furthermore, for every $0\le i \le l$, we have 
		\[
		\Pr{C_i}=1/n.
		\]
	\end{lem}	
	\begin{proof}
		Let us fix an arbitrary  $0\le i\le l$ and any $v\in V(G)$.
		Since $G$ has girth at least $\omega(l)$ and locally looks like a  $d$-regular tree, we can easily compute the number of $l$-walks visiting  $v$ in the  $i$-th step, that is 
		\[
		d(d-1)^{i-2}\times (d-1)^{l-(i-1)}=d(d-1)^{l-1}.
		\]
		On the other hand in each round, $\mathcal{A}(G, l)$~picks an $l$-walk  randomly from $nd(d-1)^{l-1}$ possible $l$-walks. Thus, we get \[
		\Pr{C_i}=\frac{d(d-1)^{l-1}}{nd(d-1)^{l-1}}=\frac{1}{n}\]
		and
		\begin{align*}
		\sum_{i=0}^{l}\Pr{C_i}=\sum_{i=0}^{l}\frac{1}{n}=\frac{l+1}{n}.
		\end{align*}
	\end{proof}		
	\begin{lem}\label{lem:Upper1}
		Suppose that  with probability $1-o(n^{-2})$, for every \mbox{$u\in V(G)$}, \[{\Empty_{t}(N(u))\geq {|N(u)|}/{2}=d/2}.\]
		Then for every $v\in V(G)$,
		\[
		\Pr{\text{ball ${t+1}$ is allocated on $v$ by $\mathcal{A}$  }} \leq \frac{\alpha}{n},
		\]
		where $\alpha$ is a constant.
	\end{lem}
	\begin{proof}
		Let  $\mathcal{E}_{ t+1,v}$  be the event that ball $t+1$ is placed on a given node $v\in V(G)$ and $\mathcal{F}_{t+1}$  be the event that 
		at least ${l}/{10}$ of nodes in
		$W_{t+1}$ are empty. Let $\neg \mathcal{F}_{t+1}$ denote the negation  of  $ \mathcal{F}_{t+1}$. Using the law of total probability,  for every $v\in V(G)$ we have
		\begin{align*}
		\Pr{\mathcal{E}_{t+1,v}}&=\underbrace{\Pr{\mathcal{E}_{t+1,v} | v\notin W_{t+1} }\cdot \Pr{v\notin W_{t+1}}}_{=0 }\notag\\
		&+\underbrace{ \Pr{\mathcal{E}_{t+1,v} | v\in W_{t+1} ~\text{and}~ \mathcal{F}_{t+1} }\cdot\Pr{v\in W_{t+1} ~\text{and}~\mathcal{F}_{t+1} }}_{\leq(10/l)\Pr{v\in W_{t+1}}} \notag\\
		&+ \underbrace{\Pr{\mathcal{E}_{ t+1, u} | v\in W_{t+1} ~\text{and}~ \neg \mathcal{F}_{t+1} }}_{\le 1}\cdot \Pr{v\in W_{t+1} ~\text{and}~ \neg \mathcal{F}_{t+1}}\notag\\
		&\le \frac{10}{l}\cdot\Pr{v\in W_{t+1}}+  \Pr{v\in W_{t+1} ~\text{and}~ \neg \mathcal{F}_{t+1}}			\end{align*}
		where  
		the first summand follows since
		if $v\notin W_{t+1}$, then ball $t+1$ cannot be placed on $v$ and  the second one  follows because
		ties are broken uniformly at random. Now, by applying Lemma \ref{first} and 
		and  Bayse' rule we have,
		\begin{align}\label{main}
		\Pr{\mathcal{E}_{t+1,v}}&\le  \frac{10(l+1)}{ln} + 
		\sum_{i=0}^l \Pr{~ \neg \mathcal{F}_{t+1} \text{and }C_i }\notag\\
		&=  \frac{10 (l+1)}{ln} + 
		\sum_{i=0}^l \Pr{~ \neg \mathcal{F}_{t+1}~|~C_i }(1/n).
		\end{align}
		In what follows we will show that for every $i$,
		\[
		\Pr{ \neg \mathcal{F}_{t+1} | C_i}\le 12/l.
		\]
		Plugging  the above bound  in Inequity (\ref{main}) 
		yields that  
		for every $v\in V(G)$,
		\[
		\Pr{\mathcal{E}_{t+1, v}}\le \frac{22 (l+1)}{ln},
		\] 
		where $22(l+1)/l$ is indeed  a constant and hence the statement is proved.

		Conditioning on event $C_i$, we only know that node $v$ is the $i$-th visited node in $W_{t+1}$, for some $0\le i\le l$. 
		Clearly, $W_{t+1}$ can be viewed as the union of  two edge-disjoint NBRWs  of lengths
		$(i-1)$ and   $l-(i-1)$ started from $v$, namely
		$W_v^1$ and $W_v^2$.
		Without loss of generality, assume that $|V(W^1_{v})|=s\ge 2$ and let \[
		W^1_{v}=(v=u_1, u_2,\ldots, u_s),\]
		where $ d(v, u_j)<d(v, u_{j'})$ for every  $1<j<j'\leq s$. 
		Clearly,
		every $ u_j\in W^1_{v}$, $2\leq j\leq s$,  is randomly chosen from   a subset of $N({u_{j-1}})$,  say $S_j\subseteq N(u_{j-1}) $ (because we run a NBRW  from $u_{j-1}$ to reach $u_j$). If it happens that the NBRW has already traversed   edge  $\{w, u_{j-1}\}$,  for some node $w$, then the walk cannot take this edge  again  and hence   
		\[
		|S_j|\in\{d, d-1\}.
		\]
		Let us  define an indicator random variable $X_{u_j}$ for every $u_j$, $2\le j\le s$, which takes one whenever $u_j$ is empty and zero otherwise.
		Thus we have 
		\[
		\Pr{X_{u_j}=1}=\frac{\Empty(S_j)}{|S_j|}.
		\]

		Let $\mathcal{K}_j$, $2\le j\le s$,  denote the event that 
		the number of empty nodes of $N({u_{j-1}})$ is at least ${d}/{2}$.	By the assumption, we have $\Pr{\mathcal{K}_j}=1-o(n^{-2})$. So,  for every $ u_j$, $2\le j\le s$, we get

		
		
		\begin{align*}
		\Pr{X_{u_j}=1 }&=\Pr{X_{u_j}=1~|~ \mathcal{K}_j}\Pr{\mathcal{K}_j}+\Pr{X_{u_j}=1 ~|~ \neg \mathcal{K}_j}\Pr{\neg \mathcal{K}_j}\\
		&\ge 1/2((d-2)/(d-1))(1-o(n^{-2}))+o(n^{-2})\ge{1}/{3},
		\end{align*}
		where the first inequality follows from \[\frac{\Empty(S_j)}{|S_j|}\ge \min \left\{\frac{d/2}{d},\frac{d/2-1}{d-1}\right\}=\frac{d-2}{2(d-1)}\]
		Since the above lower bound is independent of any $X_{u_j}$, $2\le j'\le j$, we have that for every $2\le j\le s$,
		\[
		\Pr{X_{u_j}=1~|~ X_{u_1}=x_1, \cdots,X_{u_{j-1}}=x_{j-1}}\ge 1/3.
		\]
		
		A similar argument also works for the nodes visited by $W^2_{t+1}$ and  we get $\Pr{X_u=1}\ge 1/3$, for every $u\in W_{t+1}\setminus\{v\}$. 
		Let $Y=\sum_{u\in W_{t+1}\setminus\{v\}}X_u$  be the number of empty nodes in $W_{t+1}\setminus\{v\}$. Then, we have that  $\Ex{Y}\ge l/3$. Let $Y^*$ be the summation of $l$  independent Bernoulli random variables with success probability $1/3$. 
		By applying Lemma \ref{mod-chernoff}  we get,
		{ \begin{align*}
			\Pr{ \neg \mathcal{F}_{t+1} | C_i}&\le \Pr{Y< l/6}\\
			&\le \Pr{Y^*<\Ex{Y^*}/2} \le \Pr{|Y^*-\Ex{Y^*}|\ge \Ex{Y^*}/2}.
			\end{align*}}
		We know that  $\Var{Y^*}\leq \Ex{Y^*}$, so applying  Chebychev's bound results into
		\[
		\Pr{|Y^*-\Ex{Y^*}|\ge \Ex{Y^*}/2}\leq\frac{\Var{Y^*}}{(\Ex{Y^*}/2)^2}\le \frac{4}{\Ex{Y^*}}.
		\]
		Thus, we get 
		\begin{align*}
		\Pr{ \neg \mathcal{F}_{t+1} | C_i }\leq {4 }/{\Ex{Y^*}}\leq 12/l.
		\end{align*}
	
	\end{proof}
	
	In order to prove our key lemma, we apply a  potential function argument which is similar to \cite[Theorem 1.4 ]{BSS013}.
	
	\begin{proof}[Proof of  Key Lemma \ref{lem:uniform}]
		Let us define potential function 
		\[
		\Phi(t)=\sum_{u\in V(G)}\exp({ a_t(u)}),
		\]
		where $a_t(u)$ denotes  the number of nonempty nodes of $N(u)$ after allocating $t$ balls. 
		It is clear that $\Phi(0)=n$.
		Let us assume that after allocating $t$ balls we have  $\Phi(t)\le n\cdot e^{d/4}$.
		
		\[
		\mathrm{e}^{ a_t(u)} \le \Phi(t)\le \mathrm{e}^{\log n+ d/4}.
		\]
		
		Since $d=\omega(
		\log n)$, 
		we get \[
		a_t(u)\le \log n+ d/4 < {d}/{2}
		\] and consequently 
		$\Empty_{t}(N(u))\ge \frac{d}{2}$,  for every $u\in V(G)$.  Let us define   indicator random variable $I_{t+1}(u)$ for every $u\in V(G)$ as follows:
		\[
		I_{t+1}(u):= \left\{
		\begin{array}{l l}
		1 & \quad \text{if ball $t+1$ is placed on an empty node in $N(u)$, }\\
		0 & \quad \text{otherwise.}\\
		\end{array} \right.
		\]
		Applying Lemma \ref{lem:Upper1}  shows that  if   $\Empty_t(N(u))\ge \frac{d}{2}$, then for every $u\in V(G)$,  
		\[
		\Pr{I_{t+1}(u)=1}\leq \frac{\alpha\cdot \Empty_{{t}}(N(u))}{n}\leq \frac{\alpha \cdot d}{n},
		\]
		where $\alpha$ is a constant.
		So we get
		\begin{align*}
		&\Ex{\Phi({t+1})~|~ \Phi({t})\leq n\cdot e^{d/4} }\\
		&\leq \sum_{u\in V(G)} \left\{\Pr{I_{t+1}(u)=1}\cdot e^{a_t(u)+1}+ \Pr{I_{t+1}(u)=0}\cdot e^{a_t(u)}\right\}\\
		&\leq \sum_{u\in V(G)} \left(1+\frac{ \alpha\cdot\mathrm{e}\cdot d}{n}\right)\cdot e^{ a_t(u)} =\left(1+ \frac{\alpha\cdot\mathrm{e}\cdot d}{n}\right)\Phi(t).
		\end{align*}
		Let us define  $\Psi(t):=\min\{\Phi(t), n\cdot e^{\Delta/4}\}$.  By using above recursive inequality  we have that 
		\[
		\Ex{\Psi({t+1})}\le \left(1+ \frac{ \alpha\cdot \mathrm{e}\cdot  d}{n}\right)\Psi(t).
		\]
		Thus, inductively we have that  $\Ex{\Psi({t})}\le \left(1+ \frac{ \alpha\cdot \mathrm{e}\cdot d}{n}\right)^{t}\Psi(0) $.
		Let us define \mbox{$n_1={n}/({6\mathrm{e}\alpha})$. }Then applying Markov's inequality implies that
		\[
		\Pr{\Psi({n_1})\ge n\cdot e^{d/4} }\le   \frac{ \left(1+ \frac{\alpha\cdot \mathrm{e}\cdot d}{n}\right)^{n_1}}{e^{d/4}} \le e^{-d/12}
		\]
		So with  probability $1-n^{-\omega(1)}$, we have $\Phi(n_1)=\Psi(n_1)< n\cdot e^{d/4}$. Since $\Phi(t)$ is an increasing function in $t$,  we have that 
		$\Phi(t)\le n\cdot e^{d/4}$, for every $0\le t \le n_1 $, and hence with probability $1-o(n^{-2})$, for every $u\in V(G)$,
		\[
		\Empty_t(N(u))\ge {d}/{2}.
		\]
		So, applying   Lemma \ref{lem:Upper1} shows that 
		for every $0\le t\le n_1$ and $u\in V(G)$, 
		\[
		\Pr{\text{ball $t+1$ is placed on $u$ by $\mathcal{A}(G, l)$}}\le \frac{\alpha}{n}.
		\]
		
	\end{proof}

\subsection{Proof of Lemma \ref{lem:tree}}\label{sec:tree}

	\begin{proof}
		Let us fix an arbitrary  $(\lambda, \mu)$-tree  $T\subseteq\mathcal{I}(G, l)$ and  $p_1$ be the probability that  using $\lambda$ balls  $T$ is built and  contained in $\mathcal{H}_{n_1}$. 
		There are at most $n_1\le n$ ways to choose one ball per vertex of $T$ and hence at most  $n^\lambda$ ways to choose $\lambda$ balls that are going to pick the vertices of  $T$. On the other hand, every ball picks a given vertex of $T$   with probability $1/|V(\mathcal{I})|$.
		Thus we get, 
		\[
		p_1\le n^\lambda\cdot(1/V(\mathcal{I}))^\lambda.
		\]
		Now, we have to add $c$ additional balls for very node 
		in $\cup_{W\in V(T)}W$, where $|\cup_{W\in V(T)}W|=\mu+q$, for some integer $q\ge 0$. Let $p_2$ denote  the probability that  such a event happens.  Since   $\mathcal{A}(G, l)$~is $(\alpha, n_1)$-uniform with $n_1=\lfloor n/(6\mathrm{e}\alpha)  \rfloor$, we get 
		\begin{align*}
		p_2&\le \sum_{q=0}^{\infty}{n_1\choose c\cdot(\mu+q)} \left(\frac{\alpha\cdot (\mu+q)}{n}\right)^{c\cdot(\mu+q)}\\
		&\le\sum_{q=0}^{\infty}\left( \frac{\mathrm{e}\cdot n_1}{ c\cdot (\mu+q)}\right)^{c\cdot(\mu+q)}\cdot  \left(\frac{\alpha\cdot (\mu+q)}{n}\right)^{c\cdot(\mu+q)}\\ 
		&\le \sum_{q=0}^{\infty} \left(\frac{n_1\cdot\alpha\cdot\mathrm{e}}{n\cdot c}\right)^{c\cdot(\mu+q)}=
		(1/6c)^{c\mu}{\sum_{q=0}^\infty(1/6c)^{cq}}\\
		&\le 2\cdot (1/6c)^{c\mu},
		\end{align*}
		where we use the fact that for integers $1\le a \le b $, ${b\choose a}\le (\frac{\mathrm{e}b}{a})^a$ and the last inequality follows from $\sum_{q=0}^\infty(1/6c)^{cq}\le 2$.
		Since balls are mutually  independent,  $p_1\cdot p_2$ is  an upper bound for
		the probability that   $c$-loaded $(\lambda, \mu)$-tree $T$ appears in $\mathcal{H}_{n_1}$. 
		By Corollary \ref{cor:count} we have an upper bound for the size of family of all $(\lambda, \mu)$-trees. Hence, taking the union bound over all  $(\lambda, \mu)$-trees  gives an upper bound for appearance probability of  a $c$-loaded $(\lambda, \mu)$-tree in $\mathcal{H}_{n_1}$. Thus we get,
		
		\begin{align*}
		4^\lambda |V(\mathcal{I})|\cdot \Delta^{\lambda-1}\cdot p_1\cdot p_2&\leq 2\cdot 4^\lambda |V(\mathcal{I})|\cdot \Delta^{\lambda-1}\left(\frac{n}{V(\mathcal{I})}\right)^\lambda\cdot\left(\frac{1}{6c}\right)^{c\cdot\mu}\\
		&\leq 2n\cdot 4^\lambda \cdot \left(\frac{\Delta(\mathcal{I})}{V(\mathcal{I})}\right)^{\lambda-1}\cdot\left(\frac{1}{6\cdot c}\right)^{c\cdot\mu}.
		\end{align*}
		By Lemma \ref{int} we have  $|V(\mathcal{I})|=nd(d-1)^{l-1}/2$, $\Delta(\mathcal{I})\le (l+1)^2d(d-1)^{l-1}$. So the above bound  is simplified as follows, 
		\begin{align*}
		2n\cdot4^\lambda\left(2(l+1)^2\right)^{\lambda-1}\left(\frac{1}{6}\right)^{c\cdot\mu}\le n(l+1)^{4\lambda} 6^{-c\mu} \le n\exp({4\lambda\log(l+1)-c\mu}),
		\end{align*}
		where the first inequality follows from 
		
		\[
		2\cdot 4^\lambda \cdot 2^{\lambda-1}(l+1)^{2(\lambda-1)}=8^{\lambda}(l+1)^{2(\lambda-1)}\le (l+1)^{4\lambda},
		\]
		which is true for every  $l\ge2$.
	\end{proof}

 \section{Proof of Lemma \ref{lem:PB}}
 \label{Appendix-2}

 
 
	\begin{proof}
		Let us  fix an arbitrary subpath \mbox{$P_i=[u_i, u_{i+1}]\in\mathcal{P}_k(W)$} and partition it into $3$ edge-disjoint subpaths, say   $P_i=[u_i, u]\cup_g[u, v]\cup_g[v, u_{i+1}]$,  such that  
		
		\[
		d(u_i, u)=d(v,u_{i+1})=\delta,
		\]
		where
		\mbox{$\delta=\lfloor \lfloor l/k \rfloor/4\rfloor$}.
		By the {\sf Partition} step in $PB(\rho, k)$, we know that\[
		d(u_i, u_{i+1})\in\{\lfloor l/k\rfloor,\lceil l/k\rceil\}.\] 
		So we have  \[
		d(u, v)=d(u_i, u_{i+1})- 2\delta \ge 4\delta-2\delta=2\delta.
		\]
		Let \mbox{$S=W\cap V([u, v])$} and 	
		$B(S)$ denotes the set of all balls allocated on nodes of $S$ at  height at least \mbox{$f(W)-\rho$}.
		Let $W_t$ denote   the chosen $l$-walk by ball $t\in B(S)$. Since each ball $t\in B(S)$ was allocated on a node of $S$ at height at least $f(W)-\rho$, we have that 
		$W_t$ intersects $[u, v]\subset P_i$,   $f(W_t)\ge f(W)-\rho$ and 
		\mbox{$W_t\cap {W}\neq \emptyset$}. So each $W_t$ satisfies (C2).
		Now, among all $W_t$, $t\in B(S)$, we find an $l$-walk that satisfies (C1) as well.
		We have
		\[
		|S|\ge 2\delta,
		\] 
		Every node in $S$ has load at least $f(W)\ge \rho+1$. So, every node in $S$ has  at least $\rho$ balls at height at least  $f(W)-\rho\ge 1$. Therefor we have, 
		
		
		\begin{align*}
		|B(S)|\ge |S|\rho \ge(2\delta)\rho&\ge (2\delta)(6\log_{d-1} n /\delta^2)\\
		&=12\log_{d-1} n/\delta.
		\end{align*}
		By using the above inequality we have,
		\[
		|\{W_t, t\in B(S)\}|=|B(S)|\ge 12\log_{d-1} n/\delta.
		\]
		Recall that $P_i=[u_i, u]\cup[u,v]\cup[v, u_{i+1}]$. 
		If for some $t\in B(S)$, $W_t$ contains  $u_{i}$ (or $u_{i+1}$), then it also contains subpath $[u_i, u]$ (or $[v, u_{i+1}]$), because  $W_t$ intersects $[u, v]$ and  $G$ has girth $\omega(l)$. 
		Conditioning on $\mathcal{N}_\delta$,  $[u_i,u]$ and  $[v,u_{i+1}]$ are contained in  less than  $12\log_{d-1}n/\delta$ $l$-walks. 
		So  the above inequality shows that there is at least one  ball, say $t_0\in B(S)$,  whose corresponding $l$-walk $W_{t_0}$ contains neither  $u_i$ nor $u_{i+1}$ and thus it  satisfies $(C1)$.
		Therefore we conclude that,  
		for each $P_i\in \mathcal{P}_k(W)$,
		$W_{P_i}$ exists and  $PB(\rho, k)$ on $W$ is valid.
	\end{proof}

	\section{Proof of Lemma \ref{lem:witness}} \label{app:const}
	Before we present the proof of Lemma \ref{lem:witness}, we need to show  some  lemmas about the properties of the recursive construction of the witness tree.
	Suppose that  $H_j\subset G$,  $0\le j\le h$,
	be the graphical union of all  $l$-walks  up to $j+1$-the level (i.e., $\mathcal{L}_{j+1}$). Then we have the following lemma.
	\begin{lem}\label{lem:obs}
		If $G$ has girth at least $10hl$, then, for every $0\le j\le h$, $H_j$ is a tree.
	\end{lem}
	\begin{proof}
		When $j=0$, clearly $H_0=R$, where $R$ is the root. So the diameter of $H_0$ is $l$. Assume that for some $j_0$, $0\le j_0< h$, the diameter of $H_{j_0}$ is at most $(2j_0+1)l$.  We know that
		every  $l$-walk in the $(j_0+1)$-th level intersects a path in $H_{j_0}$ so the distance between any two nodes of $H_{j_0+1}$ increases by at most $2l$ and thus the diameter of $H_{j_0+1}$ is at most 
		\[
		(2j_0+1)l+2l=(2(j_0+1)+1)l. 
		\]
		So we inductively  conclude   that,  for every $0\le j\le h$,   $H_j$  has diameter at most  $(2j+1)l$. If for some $j$, $0\le j\le h$, $H_j$ contains a cycle, then the length of the cycle is at most $2\cdot{\tt diam}(H_i)\le 2(2j+1)l\leq6hl$ which contradicts the fact that $H_j\subset G$ and $G$ has girth at least $10hl$.
	\end{proof}
	\begin{lem}\label{lem:lastlevel}
		For every $1\le j\le h$, the $j$-th level contains $k(k-2)^{j-1}$ disjoint  $l$-walks. Moreover every $l$-walk in the $j$-the level  only intersects one $l$-walk in the previous levels, which is its father. 
	\end{lem}
	\begin{proof}
		Let us begin with $j=1$. For the sake of  a contradiction, assume that $W_{P_i}, W_{P_{i'}}\in \mathcal{L}_{1}$ intersect  each other. Recall that the $l$-walks   are created by the {\sf Branch} step over the edge-disjoint paths, say   $P_{i}=[u_{i_1}, u_{i+1}]$ and 
		$P_{i'}=[u_{i'}, u_{i'+1}] \in \mathcal{P}_k(R)$. Clearly,  
		$W_{P_i}\cup_g W_{P_{i'}}$ is a connected graph as they intersect each other.
		So, by  Condition (C1), in the {\sf Partition-Branch} procedure,  we    choose two arbitrary nodes $z\in V(P_{i})\cap W_{P_{i}}$
		and $z'\in V(P_{i'})\cap W_{P_{i'}}$.
		Also,  let $\{u_{i}, u_{i+1}\}$ and $\{u_{i'}, u_{i'+1}\}$  be the {\it boundary} of $P_i$ and $P_{i'}$, respectively.
		Since $H_0$ is a tree, there is a unique path,  say $Q_{z, z'}$, in $H_0=R$ connecting $z$ to $z'$. Nodes $z$ and $z'$ have degree  $2$ in $H_0$, so $Q_{z, z'}$ contains  nodes from  boundaries of $P_i$ and $P_{i'}$.
		By (C1),   $W_{P_{i}}$ and $W_{P_{i'}}$ excludes  the boundaries. Thus we get a path from $z$ to $z'$ via $W_{P_i}\cup_g W_{P_{i'}}\subset {H_1}$ that excludes the boundaries. 
		This contradicts  the fact that there is a unique path in $H_1\supset H_0$, because $H_1$ is a tree by Lemma
		\ref{lem:obs}.
		So we infer that there are $k$ disjoint $l$-walks  in $\mathcal{L}_1$
		and they only intersect their father (i.e., $R$). 
		Recall that   $P\in \mathcal{P}_k(W)$ is a {\it free} subpath  if it does not share any node with $W$'s father. Since $W$'s, $W\in \mathcal{L}_1$, are mutually disjoint,   the nodes contained in the set of free subpaths, $\mathcal{P}_k^0(W)$, for each $W \in \mathcal{L}_1$, have degree at most  $2$ in $H_1$,
		which we call it   $\mathcal{D}_{1}$ property. In other word, $\mathcal{D}_{1}$ property says that any path in $H_{1}$ between   nodes of two free subpaths in the first level  includes nodes from boundaries of the subpaths (see Figure \ref{fig2}).
		
		Suppose  that for some $j_0$, $1\le j_0\le h$, the  statement of the lemma
		and $\mathcal{D}_{j_0}$  hold. Then we show them  for the next level as well. 
		
		Similar to case $j=1$, toward a contradiction assume that two $l$-walks $W_P, W_{P'} \in \mathcal{L}_{j_0+1}$ intersect each other. Then,  by $(C2)$ we get a path in $W_P\cup_g W_{P'}\subset H_{j_0+1}$ excluding the boundaries of 
		$P$ and $P'$ that connects  a node, say  $x$, from $P$ to another node, say $y$, in $P'$.
		By $\mathcal{D}_{j_0}$ property, the path in $H_{j_0}$ connecting $x$ to $y$  uses nodes from the boundaries, while we get a path in $H_{j_0+1}$ that exclude boundaries.
		This is a contradiction because $H_{j_0+1}\supset H_{j_0}$ is a tree
		by Lemma \ref{lem:obs}.
		So the $l$-walls in $\mathcal{L}_{j_0+1}$ are disjoint and by the construction we have $|\mathcal{L}_{j_0+1}|=(k-2)|\mathcal{L}_{j_0}|$ and hence $|\mathcal{L}_{j_0+1}|=k(k-2)^{j_0}$.
		It only remains to prove every $l$-walk only intersect its father in previous levels.
		Toward a contradiction  assume that   $W_P\in \mathcal{L}_{j_0+1}$ intersects a path, say $W$, in previous levels that is not  its father. 
		Let 
		$z'\in W_P\cap W$ and 
		$z\in W_P\cap V(P)\subset V(P)$ where 
		$P=[u, v]\in \mathcal{P}_k^0(W')$ and $W'$ is the father of $W_P$. By $(C2)$, $z$ is neither $u$ nor $v$.
		We now get a new path from  $z$ to $z'$ in $H_{j_0+1}$  excluding $u$ and $v$ (via $W_p\cup_g W$) that  contradicts the fact that there is only one path from $z$ to $z'$ in $H_{j_0}$ including a node from the boundary of $P$, as  $\mathcal{D}_{j_0}$ property holds.
		We showed that every two $l$-walks in $\mathcal{L}_{j_0+1}$ are disjoint, so the $\mathcal{D}_{j_0+1}$ holds as well.
	\end{proof}

	\begin{proof}[Proof of Lemma \ref{lem:witness}]
		Let us consider a graph $T$ whose nodes are the set of all  $l$-walks in $\bigcup _{j=0}^h \mathcal{L}_j$, where $\mathcal{L}_0=\{R\}$. And  two nodes are connected, if and only if  the corresponding  $l$-walks intersect each other or vice versa. By Lemma  \ref{lem:lastlevel}  for every $1\le j\le h$, the $j$-th,   level contains $k(k-2)^{j-1}$ disjoint $l$-walks and they  intersect either their fathers or their $k-2$ children.
		This implies that $T$ is a subtree of interference graph $\mathcal{H}_{n_1}$ with 
		\[
		|V(T)|=\lambda=1+ k\sum_{j=0}^{h-1} (k-2)^j.
		\]
		If we only consider the $h$-th level, then we get 
		\[
		\left|\cup_{W\in V(T)}W\right|\ge \mu=(l+1)\cdot k(k-2)^{h-1}.
		\]
		By  (C2) in the {\sf Partition-Branch}  procedure we have  that 	for every $W\in \mathcal{L}_j$, $1\le j\le h$,
		
		\[
		f(W)\ge (h-j)\rho+c.
		\] 
		Hence every node in $\cup_{W\in V(T)}W$ has load at least $c$.
	\end{proof}
	
\section { Proof of Lemma \ref{lem:Max}}\label{Max-1}
\begin{proof}
	Let us fix an arbitrary path $[u, v]$ of length $\delta$, where \[
	\delta=\lfloor\lfloor l/k\rfloor/4\rfloor=
	\min\{\lfloor l/16\rfloor,\lfloor \sqrt{\log_d n}/4\rfloor\}.
	\]
	The latter equality is true  because we set 	$k=\max\{4, \lfloor l/\sqrt{\log_d n}\rfloor \}$.
	Clearly, if $W$ be an $l$-walk and $[u, v]\subseteq W=[u_0, u_l]$, then
	\[
	d(u_0, u)+d(v, u_l)=l-\delta.
	\]
	Moreover,  $G$ is  a $d$-regular graph  with  girth at least $\omega(l)$, so 
	the total number of different paths of length $l$ containing $[u,v]$ is 
	\[
	\sum_{a+b=l-\delta}{(d-1)^{a}(d-1)^{b}}=(l-\delta+1)\cdot (d-1)^{l-\delta}.
	\] On the other hand the total number of different paths of length $l$ is 
	$n\cdot d\cdot (d-1)^{l-1}/2$. So the probability that in some round $t$, $1\le t\le n$, we get  $[u,v]\subseteq{W}_t$ is at most 
	\[
	\frac{2(l-\delta+1)(d-1)^{l-\delta}}{n\cdot d\cdot(d-1)^{l-1}}= \frac{2(l-\delta+1)(d-1)}{n \cdot d\cdot (d-1)^{\delta}} \le \frac{2l}{n(d-1)^{\delta}}.
	\]
	Let $u_\delta=\lceil6\log_{d-1} n/\delta\rceil$ and $\{t_1, t_2,
	\ldots, t_{u_\delta}\}\subset[n]$ be  a sequence of distinct rounds  of size $u_\delta$. We define indicator random variable  
	$X_{t_1, t_2,\ldots, t_{u_\delta}}([u, v])$, which
	takes one if 
	$[u, v]\subseteq W_{t_i}$, for every $1\le i\le u_\delta$, and zero otherwise. Thus we get
	\begin{align*}
	\Pr{X_{t_1, t_2
			\ldots, t_{u_\delta}}([u, v])=1}&\leq\left({2l}/{n(d-1)^{\delta}}\right)^{u_\delta}\\
	&= n^{-u_\delta} \left({d-1}\right)^{(\log_{d-1}(2l)-\delta)u_\delta}\\
	& \le n^{-u_\delta}{{(d-1)}}^{-{u_\delta\cdot \delta}/2}= n^{-u_\delta}n^{-3},
	\end{align*}
	where the last inequality follows from  $l\in [\omega(1), o(\log n)]$ and hence,  
	\[
	\log_{d-1} (2l)\leq \delta/2.
	\]
	
	There are at most $n^{u_\delta}$  sequences of rounds of size $u_\delta$ and at most \mbox{$n(d-1)^{\delta-1}$}
	paths of length $\delta$.
	Thus,  by using the previous upper bound and the union bound over all sequences of rounds and paths of length $\delta$ we have
	
	\begin{align*}
	&\sum_{\delta\text{-path}}\sum_{t_1, t_2,\ldots, t_{u_\delta}}\Pr{X_{t_1, t_2
			\ldots, t_{u_\delta}}([u, v])=1}\\
	&\leq n d (d-1)^{\delta-1}n^{u_\delta}
	\Pr{X_{t_1, t_2\ldots, t_{u_\delta}}([u, v])=1}\\
	&\le o(n^{2})n^{u_\delta} \Pr{X_{t_1, t_2\ldots, t_{u_\delta}}([u, v])=1}=o(1/n),
	\end{align*}
	where the last inequality  follows from 
	$\delta\le l=o(\log_{d}n)$.
	This implies that
	with probability $1-o(1/n)$
	there  is no path of length $\delta$ contained in at least $u_\delta$ $l$-walks or equivalently 
	$\mathcal{N}_\delta$ holds.
\end{proof}

	\section{Balanced Allocation on Sparse Graphs}\label{sec:sparse}
	In this section we present allocation algorithm $\mathcal{A}'(G, l)$ 
	for $d$-regular graphs, with $d\in[3, \Oh(\log n)]$.
	The algorithm proceeds as follows:
	In each round, every ball picks a node uniformly at random and it takes a NBRW of length $l\cdot r_G$ from the chosen node, where $r_G=\lceil 2\log_{d-1}{\log n}\rceil$. After that the ball collects the load information every $r_G$-th visited node, called the {\it potential choice}, and place itself on a least-loaded potential choice (ties are broken randomly).
	We now present  the following theorem for the maximum load attained by  $\mathcal{A}'(G, l)$.
	
	\begin{thm}\label{thm:sparse}
		Suppose that $G$ is a $d$-regular graph with girth at least $10l(\log \log n)^2$ and $d\in[3, \Oh(\log n)]$. Then, 	with high probability the maximum load attained by $\mathcal{A}'(G, l)$, denoted by $m^*$,  is bounded from above as follows:
		\begin{itemize}
			\item[I.] If $\omega(1)\le  l\le 4\gamma'_G$, where $\gamma'_G=\sqrt{\log _d n/r_G}$, then we have
			\[
			m^*=\Oh\left(\frac{ \log_d n\cdot \log\log n}{r_Gl^2}\right).\]
			
			\item[II.] If  $l\ge 4\gamma'_G$, then we have 
			\[
			m^*=\Oh\left(\frac{\log\log n}{\log (l/\gamma_G)}\right).
			\]
		\end{itemize}
		
	\end{thm}
	
	The analysis of  allocation algorithm $\mathcal{A}'(G, l)$ is almost the same as the algorithm for dense graphs so we only outline  parts of  the proof and notations that are slightly different.
	Let us start by defining an interference graph and show some of its properties (similar to Section \ref{sec:defi})
	\begin{defi}(Interference graph )
		For any given pair $(G, l)$, interference graph  $\mathcal{I}'(G,l)$ is a graph whose nodes are  the set of  potential choices in each $l\cdot r_G$-walk on  $G$ and two sets are connected if and only if they intersect each other.
	\end{defi}
	\begin{lem} We have that
		\begin{align*}
		&~|V(\mathcal{I}')|=nd(d-1)^{l\cdot r_G-1}/2,\\
		&~\Delta(\mathcal{I}')\le
		(l+1)^2d(d-1)^{l\cdot r_G-1}.
		\end{align*} 
	\end{lem}
	\begin{proof}
		Since $G$ has girth at least $l\cdot r_G$, each $l\cdot r_G$-walk determines a unique set of potential choices and hence $|V(\mathcal{I}')|$ is the number of all $lr_G$-walks, that is,
		\[|V(\mathcal{I}')|=nd(d-1)^{l\cdot r_G-1}/2.
		\]
		It is easy to show that, 
		 for every $0\le i\le l$, each node can be the $(ir_G)$-th visited  node of $d(d-1)^{lr_G-1}$ $lr_G$-walks. So every node is contained in at most $(l+1)d(d-1)^{lr_G-1}$ many walks. On the other hand, each set contains $l+1$ nodes and hence, 
		\[
		\Delta(\mathcal{I}')\le
		(l+1)^2d(d-1)^{l\cdot r_G-1}.
		\]
	\end{proof}
	
	Similar to Section \ref {sec:defi}, we can define the 
	$(\lambda, \mu)$-trees contained in  $\mathcal{I'}$ and get the similar results in terms of $V(\mathcal{I}')$ and $\Delta(\mathcal{I}')$.

	The shown results for dense graph including appearance probability of a $c$-loaded $(\lambda, \mu)$-tree (Section \ref{sec:defi}) and witness graphs (Section \ref{sec:witness})
	are based on the local properties of a regualr graph with degree at $\omega(\log n)$ and girth $10l h$. We also know that 	   
	allocation algorithm $\mathcal{A}'(G, l)$ samples nodes from a neighborhood of radius $l\cdot r_G$. It only collects the load information every $r_G$-th  visited node and ignore rest of the visited nodes. 
		Now if we sort potential choices according to their distance from the starting node, say $(u_0, u_1\cdots, u_l)$. It is easy to see that each $u_i$, $1\le i\leq l$, is chosen from a set of nodes of size at least $(d-1)^{r_G}=\omega(\log n)$, as the graph locally looks like a tree.

	Roughly speaking, $\mathcal{A}'(G, l)$ reduces sparse $d$-regular graph $G$  to a regular graph with degree $\omega(\log n)$ and hence the results also hold for $\mathcal{A}'$.

	\section{A Lower Bound}\label{App:lower}
	In this section  we derive a lower bound for the maximum load attained by allocation  algorithms  on  dense and sparse graphs.
	Define
	\[
	r_G:= \left\{
	\begin{array}{l l}
	1 & \quad \text{if $d=\omega(\log n)$ , }\\
	\lceil 2\log_{d-1}{\log n}\rceil & \quad \text{otherwise.}\\
	\end{array} \right.
	\]
	
	Consider the generic balanced allocation $\mathcal{A}(G, l)$ that proceeds af follows:
	each ball takes a NBRW of length $l\cdot r_G$, and then it places itself on a least-loaded node among every $r_G$-th visited node. When $r_G=1$, it means every visited node is considered  as a potential choice. It is easy to see that the algorithm covers both classes of graph. Now we show the following lower bound for the maximum load.

	\begin{thm}[Lower Bound]\label{thm:lower}
		Suppose that $G$ be a $d$-regular $n$-vertex  graph with  girth at least $\omega(l)$, where  $ l\in [32 r_G, \Oh(\gamma_G)]$ is an integer and 
		$\gamma_G=\sqrt{\log_{d}n/r_G}$. Then  with  probability $1-n^{-\Omega(1)}$  the maximum load  attained by $\mathcal{A}(G, l)$ is  at least \[\Omega\left(\frac{\log_{d}n}{r_G\cdot l^2}\right).
		\]
	\end{thm}
	
	\begin{proof}
		We know in each round the algorithm picks  a random path of length $l\cdot r_G$. Let us define  indicator random variable $X_{P}$ for every path of length $l\cdot r_G$  as follows,
		\[
		X_{{P}}:= \left\{
		\begin{array}{l l}
		1 & \quad \text{if $P$  is chosen at least  $\tau$ times by $\mathcal{A}$, }\\
		0 & \quad \text{otherwise,}\\
		\end{array} \right.
		\] 
		where $\tau$ will be specified later.
		The total number of paths of length $l\cdot r_G$, say $s$, is 
		\[
		nd(d-1)^{l\cdot r_G-1}/2\le nd^{l\cdot r_G}/2.
		\]
		
		Let $P$ be an arbitrary path of length $l\cdot r_G$ in $G$. Thus we get 
		\begin{align}\label{prob1}
		\Pr{X_{{P}}=1}&=\sum_{i=\tau}^n {n \choose i}\left(\frac{1}{s}\right)^i\left(1-\frac{1}{s}\right)^{n-i}
		\ge \left(\frac{n}{s\cdot \tau}\right)^\tau \left(1-\frac{1}{s}\right)^{n}\notag\\
		&\ge \left(\frac{2}{d^{lr_G}\cdot \tau}\right)^\tau\left(1-\frac{1}{s}\right)^s\ge  d^{-(lr_G+\log_d\tau)\tau}/\mathrm{e},
		\end{align}
		where the second inequality follows from 
		$n\le s\le n\cdot d^{lr_G}/2$.
		By setting 
		\[
		\tau=\frac{\log_d n}{6l\cdot r_G},
		\]
		and using the fact that
		$\log_d\tau<\log_d\log_d n\le r_G\leq l$ we get 
		\[
		(lr_G+\log_d\tau)\tau\le \log_d n/6 + {\log_d n}/6=\log_d n/3.
		\]
		By substituting  the above upper bound in (\ref{prob1}), we get
		\[
		\Pr {X_P=1}=\Omega(n^{-1/3}).
		\]
		Let us define  the random variable $Y=\sum_{{\text{ all paths}}} X_P$.
		By linearity of expectation we have   
		\begin{align}\label{infinitexp1}
		\Ex{Y}=s\cdot\Pr{X_P=1}=(n\cdot d\cdot (d-1)^{lr_G-1}/2) \Omega(n^{-1/3})=\Omega(n^{2/3}).
		\end{align}
		It is easily seen that  the random variables $X_P$ and $X_{P'}$  are negatively correlated,   which means for every  $P$ and $P'$,
		\[
		\Ex{X_P\cdot X_{P'}}\le \Ex{X_P}\cdot\Ex{X_{P'}}.
		\] This implies that 
		\begin{align*}
		\Var{Y}&=\sum_{{P}}(\Ex{X^2_{{P}}}-(\Ex{X_{{P}}})^2)+\underbrace{\sum_{{P}\neq {P}'} (\Ex{X_{{P}}X_{{P}'}}-\Ex{X_{P}}\Ex{X_{{P}'}})}_{\leq 0}\\
		&\leq \sum_{{P}} \Ex{X^2_{{P}}}=\Ex{Y}. 
		\end{align*}
		Applying Chebychev's inequality and above inequality  yield   that 
		\[
		\Pr{Y=0}\le \Pr{|Y-\Ex{Y}|\ge\Ex{Y}}=\frac{\Var{Y}}{(\Ex{Y})^2} \leq \frac{1}{\Ex{Y}}.
		\]
		
		By equality (\ref{infinitexp1}) we have that $\Ex{Y}=\Omega(n^{2/3})$. Therefore with probability at least $1-\Oh(n^{-2/3})$ we have $Y\ge 1$, which means  
		there exists a path $P$  that   is chosen at least $\tau$ times.
		Since every $P$ contains  $l+1$ choices,
		by the pigeonhole principle  there is a node  with load at least 
		\[
		\Omega\left(\frac{\tau}{l}\right)=\Omega\left(\frac{\log_{d}n}{r_Gl^2}\right).
		\]
		
	\end{proof}

\end{document}